\documentclass[aps,prb,floatfix]{revtex4}
\usepackage{amsmath, amssymb, amsthm, alltt, setspace,bbm}
\usepackage[bookmarks=true, pdfpagemode = None, pdfstartview = FitH,colorlinks =
true, citecolor = black, linkcolor = black, urlcolor=black]{hyperref}
\usepackage{multirow}
\usepackage{graphicx}
\usepackage{booktabs}

\long\def\symbolfootnote[#1]#2{\begingroup%
\def\thefootnote{\fnsymbol{footnote}}\footnote[#1]{#2}\endgroup}

\newcommand{\be}{\begin{equation}}
\newcommand{\ee}{\end{equation}}
\newcommand{\NP}{{\sf{NP}}}
\renewcommand{\P}{{\sf{P}}}
\hypersetup{
  letterpaper,
  colorlinks=false,
  urlcolor=black,
  pdfpagemode=none,
  pdftitle={The time-dilation Effect as a Computational Resource},
  pdfsubject={The time-dilation Effect as a Computational Resource},
  pdfkeywords={twin paradox, clock paradox, time-dilation, time dilation, computation, complexity theory, computability, information theory, entanglement, grover search, quantum search, computer science, special relativity, Differential aging from acceleration, general relativity}
}

\theoremstyle{change}
\newtheorem{definition}{Definition}[section]

\newtheorem{postulate}[definition]{Postulate}
\newtheorem{theorem}[definition]{Theorem}
\newtheorem{corollary}[definition]{Corollary}

\begin{document}
\title{The Computational Power of Minkowski Spacetime}
\author{Jacob D. Biamonte}\email{biamonte@wolfson.ox.ac.uk}\affiliation{Oxford University Computing Laboratory \\Wolfson Building, Parks Road\\ Oxford, OX1 3QD, United Kingdom}

\begin{abstract}
The Lorentzian length of a timelike curve connecting both endpoints of a classical computation is a function of the path taken through Minkowski spacetime.  The associated runtime difference is due to time-dilation: the phenomenon whereby an observer finds that another's physically identical ideal clock has ticked at a different rate than their own clock.  Using ideas appearing in the framework of computational complexity theory,  time-dilation is quantified as an algorithmic resource by relating relativistic energy to an $n$th order polynomial time reduction at the completion of an observer's journey. These results enable a comparison between the optimal quadratic \emph{Grover speedup} from quantum computing and an $n=2$ speedup using classical computers and relativistic effects.  The goal is not to propose a practical model of computation, but to probe the ultimate limits physics places on computation.
\end{abstract}
\maketitle

\section{Introduction}\label{sec:intro}
This work considers the following questions: i.) what is the computational power gained from utilizing time-dilation in Minkowski spacetime, and ii.) how can techniques from resource analysis motivate the derivation of relations among physical quantities, such as energy and time?

To this end, consider a classical computation requiring $N$ operations each taking time $\Delta t$, then $\Delta t N$ becomes the \textit{total runtime} in a local inertial frame.  Indeed, relativistic effects don't change complexity results inside an inertial frame, and the number $N$ is agreed on by all observers (see~\ref{prop:1}).  However, the total computational time experienced by observers in motion is relative.  We will state this as a relation between a polynomial reduction inside the black box model and relativistic mass --- this is accomplished in Section \ref{sec:mass} where Theorem~\ref{theorem:energy} is proven.
\begin{theorem}\label{theorem:energy}
The minimal relativistic energy $E$ $(= mc^2)$ required to perform a computation within time $N$ in an inertial frame and time $N^{1/n}$ in the frame of an observer $O$ is given as
\be\label{eqn:Econstant}
E  = N^{1-1/n} m_0 c^2,
\ee
where $m_0$ is the rest mass of $O$, $\Delta t :=1$, $n$ is the order of the sought polynomial reduction and $N\in\mathbb{N}^*, n\geq 1$.
\end{theorem}

We have found a non-linear tradeoff \eqref{eqn:Econstant} between relativistic energy and reduction of computational runtime $\Delta t N$.  Let us continue by stating Einstein's two postulates (\ref{prop:1} and \ref{prop:2}), translated into computer science terms\footnote{Time dilation was first predicted by Einstein~\cite{eins}.  This effect was experimentally observed in 1938 by Ives and Stilwell~\cite{ives38}, observation in macroscopic clocks occurred in 1972~\cite{experiment72} with the current state of the art found in~\cite{exptodate}.}.  The standard presentation of the postulates is readily found~\cite{rindler,woodhouse}


\begin{postulate}\label{prop:1}
The time needed to compute a function in a local inertial frame is the same for all observers in uniform motion.
\end{postulate}

\begin{postulate}\label{prop:2}
The upper bound on the speed of information is the same for all observers in uniform motion.
\end{postulate}

\textbf{Structure of this paper: }We will continue by explaining the computational model selected for this study. This is followed by Section \ref{sec:constantV}, which relates computation between observers in relative motion and develops relations between the sought polynomial reduction parameter $n$ and Bondi's $k$-factor~\cite{rindler,woodhouse}.  A computational version of the \emph{Twin Paradox} is considered in Section \ref{sec:twins}.  Before considering uniform acceleration in \ref{sec:acceleration}, Section \ref{sec:mass} considers the relativistic energy required to produce an $n$th order polynomial runtime reduction; used to prove Theorem \ref{theorem:energy}.  Before concluding, Section \ref{sec:grover} compares the quadratic \emph{Grover speedup}~\cite{grover-1997-79} from quantum computing to the $n=2$ speedup found using classical computers and relativistic effects.

\subsection{Computational Complexity Theory}\label{sec:model}
The resources consumed by an \textit{efficient algorithm} scale polynomially in the problem size --- necessarily from the class $\P{}$ (of problems known to be efficiently solvable).  The most famous open question in Computer Science concerns proving if it is impossible to efficiently solve a complete problem from the class $\NP{}$~\cite{Pap94,sipser,GJ79}.

Although the $\P{}{\neq}\NP{}$ question regards general algorithmic complexity, one can consider its physical analogue by asking if the laws of physics allow, even in principle, the existence of a physical process that can be harnessed to speed up the solution to an $\NP{}$-complete problem~\cite{aar:np,KSV02}.  Deterministic query complexity in the black-box model is the ideal framework to address this question\footnote{To date, research connecting computer science and relativity theory has been focused on the implications the existence of closed time-like curves would have on computation~\cite{Geroch:1986iu,deutsch:ctc,NTMHst02,RCTB06,brun,bacon,aaronson-2008}.  Among other interesting consequences, their existence would imply the efficient solution of $\NP{}$-complete problems.}.

Consider a classical device computing $f:\{0,1\}^n\rightarrow \{0,1\}::x\mapsto f(x)$, for a given function defined on its range of inputs for positive integer $n$.  In the case of an unstructured database, one is given $f$ in a black-box with a promise that $f$ outputs $1$ for a single input $x'$ and $0$ otherwise.  Using a classical computer, to determine $x'$ in the worst case requires $N:=2^n$ \emph{queries} of the search space --- each taking time $\Delta t$.  We will consider solving such an $\NP{}$-complete problem\footnote{As an example of what we will consider, let $(ct,R\cos(\omega t), R\sin(\omega t),0)$ represent the world curve of a particle traveling the $2\pi R= 26,659$ meters at velocity $R\omega=.999999991c$ to complete one lap around the Large Hadron Collider~\cite{lhc}. On the day this sentence was written, the library of congress contained $N$ = 32,332,832 books.  Consider an ideal computer traveling with the particle able to completely search one book record per lap.  If one left the computer on the ground while traveling with the particle, after the completion of $4,386$ laps the computer would have exhaustively searched all $32,332,832$ books --- here $n = 2.99$.  Due to the circular path, one will experience an acceleration $a$ --- it will be larger than the classical value $a=R\omega^2$ by a factor $(1-R^2\omega^2/c^2)^{-1}$.}.  

\section{Einstein's Computer of 1905}\label{sec:constantV}
Let $I$ denote an observer traveling along a geodesic that is the common origin of the coordinate system and let $O$ denote an observer in relative motion.  Consider an event \textbf{0} at which i.) $I$ and $O$ pass; ii.) synchronize their temporal and spatial orientations and iii.) the computation begins.

Keeping in mind the appropriate units, consider the general question of relating the proper time interval experienced in $O$'s frame as the $n$th root of the proper time interval experienced in $I$'s frame.  This is done by relating the clocks of $I$ and $O$ --- we insist that $t(\tau)=T^n~\text{sec}$ and $\tau(t)=T~\text{sec}$, where $O$ reaches $d$ at the event $(cT^n,d,0,0)$ in $I$'s frame.  This is made possible by the Lorentz factor $\gamma(u)$ through the temporal relation:
\be\label{ean:gamma:Tn:T}
T^n~\text{sec} = \gamma(u)T~\text{sec},
\ee
where $u$ is the velocity of $O$ measured in $I$'s frame.  For analysis purposes, consider constant $T>1,n\geq 1$ and let $T^n$ and $T$ have units of seconds only where applicable, where $T^{n-1}$ is dimensionless.  The velocity $u$ is now expressed as
\be\label{eqn:u}
u(T,n) = c\sqrt{1-T^{2-2n}},
\ee
where we choose an origin $u\geq 0$.  Now consider reparameterization of \eqref{eqn:u} with $T^{n-1}:=\cosh (\phi(u))$, it follows that $\phi(u)=\text{arctanh}\left(u/c\right)=\log (k)$ with $\phi(u)$ the \emph{rapidity}~\cite{rindler,woodhouse} and the dimensionless parameter $k$ known as Bondi's $k$-factor~\cite{rindler,woodhouse}.  The right hand side expansion of $k=\sqrt{(c+u)/(c-u)}\approx1+u/c+O\left(u/c\right)^2\geq1$ is found from letting $u\rightarrow 0$. One thus recovers the limit where classical computing holds when $u/c\ll 1$ in which case $k\approx 1$.

One relates the order of the sought reduction $n(k,T)$ to the $k$-factor as $n  =   1+\log(k/2+k^{-1}/2)/\log (T)$, where the classical regime ($n=1$) is recovered directly by setting $k=1$.  Alternatively, consider the series expansion as $k\rightarrow 1$, $n \approx   1+(k-1)^2/\log(T^2)$, and note that $n\approx 1$ for $k-1$ near zero, as expected.  In terms of the $k$-factor, $n$ scales as $O(\log k)$ to leading order.  One also finds that $k+k^{-1}=2T^{n-1}$ and that $k-k^{-1}=2\sqrt{T^{2n-2}-1}$.  Finally, the $k$-factor can be expressed in terms of the query time $k(n,T)=T^{n-1}\left(1+\sqrt{1-T^{2-2n}}\right)$, where it is again easily seen that $n=1$ yields the classical limit. In terms of the total query time, the $k$-factor scales as $O(T^{n-1})$ to leading order.

\subsection{The Twins get computers}\label{sec:twins}

In 1911, Paul Langevin made Einstein's 1905 prediction of time-dilation vivid by noting asymmetry in a thought experiment involving twins ($O$ and $I$) --- both measure events on $O$'s worldcurve.  This became known as the \emph{twin} or \emph{clock paradox} and was a subject of debate during the first half of the last century~\cite{darwin57} and remains a research area today~\cite{minguzzi-2005-73,iorio-2005-18,minguzzi-2006-19}.  For completeness, let us then state this paradox in terms of our framework.


The twins calculate the time ($T^n = \Delta t N$) needed in $I$'s frame to perform a computation with the understanding that $O$ wishes to have the solution in the $n$th root of this time ($T = (\Delta t N)^{1/n}$ sec), upon returning from a journey.    The time of the total trip measured by $I$ is $T^n=2d/u$ and by $O$ is $\gamma(u)T=2d/u$, and the spacetime path is given as (\ref{def:path1}):

\begin{definition}\label{def:path1}
{\upshape (Spacetime path \ref{def:path1})} Consider three events: {\upshape i.)} $O$ starts from rest, reaching a constant velocity $(u=c\sqrt{1-T^{2-2n}}$ from \eqref{eqn:u}$)$ within a negligibly short time leaving her twin $I$ to perform a computation for time $T^n$; {\upshape ii.)} after journeying for time $T/2$ in $O$'s frame, and some distance $d/2$ in $I$'s frame, $O$ suddenly reverses velocity; {\upshape iii.)} $O$ arrives back at her starting point, stops, and recovers the result of the computation.
\end{definition}

Consider the Lorentz transform of the event where $O$ reaches $d$, and hence recover the respective temporal and spatial relations: $cT^n = \gamma(u)Tc$ and $d=\gamma(u)Tu$.  From Postulate~\ref{prop:2} the distance $d$ must be less than $cT^n$, the relativistic limit placed on massive bodies.  One can also establish the velocity independent expression $n(d,T)$ as
\be\label{eqn:nT}
n = 1+\log\left(1+ d^2/(cT)^2\right)/\log(T^2).
\ee

It is often stated that, ``it is possible to travel as far as you like in as short a time as you like, provided the distance ($d$) is measured before you set off and the time ($T$) is measured along your world-line'' (see~\cite{taylor,woodhouse} for instance).  The expression \eqref{eqn:nT} now gives the preceding statement computational meaning in terms of a polynomial reduction $n$ inside the black-box model.

\subsection{Equivalence among the Polynomial Reduction and Relativistic Mass}\label{sec:mass}
We let $E$ represent energy measured in the frame $I$, and use $m_0$ to denote the rest mass of $O$.  We stated the equivalence among the polynomial reduction and relativistic mass in Theorem~\ref{theorem:energy} in Section~\ref{sec:intro} --- an outline of the proof follows:

\begin{proof} (Theorem~\ref{theorem:energy}) The proof relies on the results of Sections \ref{sec:constantV} and \ref{sec:twins}.  From \eqref{eqn:u} and the relation $E=mc^2$ it can be established that $E = mc^2 = T^{n-1}m_0c^2$ and Theorem~\ref{theorem:energy} follows.
\end{proof}
\begin{corollary}
From \eqref{eqn:Econstant} it follows that $n(E,T) = 1+\log\left(E/(m_0c^2)\right)/\log\left(T\right)$.
\end{corollary}

The 4-momentum of $O$ measured in the frame $I$ can now be expressed as $P=T^{n-1}(E_0/c,p_1,p_2,p_3)$, where $p_i=T^{n-1}m_0u_i$ is the $i$th component of the 3-momentum.  When $E=E_0=m_0c^2$ one recovers the classical limit $n=1$.  In Minkowski spacetime, it is when the relativistic energy of an observer $O$ increases past their rest energy in a frame $I$, that computational gains of order $n>1$ become possible.

\section{Computation by Uniform Acceleration}\label{sec:acceleration}
Now consider the case of constant acceleration ($a$) in the spatial direction $(1,0,0)$ over a distance $d$ measured by $I$.  The world-line of $O$ measured in the frame $I$ is given as\footnote{The standard approach to derive the equations of motion for the case of constant acceleration can be found in relativity books including~\cite{rindler,woodhouse,graviation}. Let overdot $(~\dot{}~)$ denote differentiation with respect to $\tau$.  The 4-vector components of the velocity $U$ and the acceleration $A$ of $O$'s motion as measured by $I$ are given as $(c\dot t,\dot x,0,0)$ and $(c\ddot t, \ddot x, 0,0)$, respectively.  One finds that $V^\mu V_\mu=c^2\dot t^2 - \dot x^2 = c^2$ and that $A^\mu A_\mu = c^2\ddot t^2 - \ddot x^2 = -a^2$.  From differentiation of $V^\mu V_\mu$ and by substituting into $A^\mu A_\mu$ one recovers $c\ddot t = a \sqrt{\dot t^2-1}$ and $\ddot x = a\dot t$.  With a suitable choice of origin one finds: $\dot t = \cosh \left(a\tau/c\right)$ and $\dot x = c\sinh \left(a\tau/c\right)$ and the result follows}
\begin{equation}\label{eqn:worldcurve}
(ct,x,y,z)={c^2}/{a}\left( \sinh \left( {a\tau}/{c}\right), \cosh \left({a\tau}/{c}\right),0,0\right),
\end{equation}
which is a hyperbola in the the $(ct,x)$-plane with asymptotes $ct = \pm x$.  The instantaneous velocity of $O$ measured in $I$'s frame becomes
\be\label{eqn:u_t}
u(\tau):=dx(\tau)/dt = c \tanh \left(a\tau/c\right)={at}/{\sqrt{1+(at/c)^2}}.
\ee
It follows that
\be\label{eqn:gamma_t}
\gamma(u(\tau))\lvert_\tau = \cosh \left(a\tau/c\right)=\sqrt{1+(at/c)^2}={ax}/{c^2}+1,
\ee
where $\gamma(u(\tau))\lvert_\tau$ is the Lorentz factor at proper time $\tau$ --- a smooth analogue of $\Delta t/\Delta \tau$ from \eqref{eqn:u}.

One could again couple the clocks of $I$ and $O$ by insisting that $t(\tau):=T^n=\Delta t N$ and $\tau(t)=T=(\Delta t N)^{1/n}$, where we slightly abuse notation by letting $(\Delta t N)^{1/n}$ and $\Delta t N$ have units of seconds; however, $\gamma$ is now a function of $\tau$ and the temporal relation \eqref{eqn:worldcurve} already provides a coupling --- it follows that\footnote{By finding $a(T,n)$ in the first of \eqref{eqn:TnasT} one recovers a transcendental equation which is solved in terms of Lambert's $W$ function~\cite{w1996}.}
\begin{gather}\label{eqn:TnasT}
(\Delta t N)^{1/n}~\text{sec}= \frac{c}{a}\sinh^{-1}\left(\frac{a}{c}\Delta t N\right)=\frac{c}{a}\cosh^{-1}\left(\frac{ad}{c^2}+1\right)~~~\text{and}~~~
\end{gather}
\be
d = \frac{c^2}{a}\left(\cosh \frac{a}{c}(\Delta t N)^{1/n}-1\right)=\frac{c^2}{a}\left(\sqrt{1+(a\Delta t N/c)^2}-1\right),
\ee
where $d:=x(T)-x(0)$.  We have recovered i.) the time dependent gamma factor $\gamma(u(\tau))\lvert_\tau$ in \eqref{eqn:gamma_t} and ii.) the instantaneous velocity in \eqref{eqn:u_t}. We can now consider the energy efficiency of computation by uniform acceleration. 

\begin{definition}\label{def:path2}
{\upshape (Spacetime path \ref{def:path2})} Let $g$ be a positive constant and consider: $O$ starting from $($returning to$)$ rest at $t=\tau =0$ $(\tau=T)$, $u(0)=0$ and the four intervals: {\upshape i.)} $a(\tau)=g$ with $\tau\in[0,T/4]$; {\upshape ii.)} $a(\tau)=-g$ with $\tau\in[T/4, T/2]$; {\upshape iii.)} $a(\tau)=g$ with $\tau\in[T/2, 3T/4]$; {\upshape iv.)} $a(\tau)=-g$ with $\tau\in[3T/4,T]$\footnote{We note that on the path (\ref{def:path2}), the maximum speed $u(0)=c \tanh \frac{aT}{4c}$ occurs at the end of the acceleration in i.) and iii.). The maximum distance $d=\frac{c^2}{a}\left(\cosh \frac{aT}{2c}-1\right)$ between $O$ and $I$ in $I$s frame is at the end interval ii.).}.
\end{definition}

\begin{theorem}\label{theorem:energya}
The total relativistic energy $E_\text{total}$ required to preform a computation by uniform acceleration on the spacetime path \ref{def:path2} scales as $O\left(N\right)$.
\end{theorem}
\begin{proof} (Theorem~\ref{theorem:energya}) Consider a single leg in the spacetime path corresponding to $1/4$th of the computation $(\Delta t N/4)$:  At the start of the path, the energy is $(M+m_0)c^2$, where $M$ is the rest mass of the fuel turned into light and used to propel $O$. Transferring the mass $M$ into light results in the energy at the end of the leg being $\gamma(u)m_0c^2 + E_l$ and hence $(M+m)c^2=\gamma m c^2 +E_l$.  The momentum at the start of the leg is zero and at the end of the leg is $\gamma(u) m u - E_l/c$ --- it follows that $M = m_0 \exp(a\tau/c)-m_0$.  Using expressions \eqref{eqn:TnasT} one finds that $M=m_0(\frac{ad}{c^2}+\sqrt{{a^2d^2}/{c^4}-1}-1)=m_0(\frac{a\Delta t N t}{c}+\sqrt{{a^2\Delta t^2N^2}/{c^2}+1}-1)$ and Theorem~\ref{theorem:energya} follows.
\end{proof}

\section{Racing a quantum computer through Minkowski Spacetime}\label{sec:grover}
Quantum query complexity broke classical lower bounds on the required number of queries and hence the total time interval ($\Delta t N$) to solve certain black-box problems including database search~\cite{procon,grover-1997-79}.  Let us examine the related speedup using classical computers together with relativistic effects.  To recover a Grover speedup~\cite{grover-1997-79} for the case of constant velocity requires energy
\be\label{eqn:Egrover}
E = \sqrt{N} m_0 c^2,
\ee
in $I$'s frame, where $m_0$ is the rest mass of $O$, $\Delta t~(:=1)$ is the single query time and $N$ gives the total number of items in a search space.

Here is how a classical computer can outperform a quantum one:  to avoid the scaling of Theorem~\ref{theorem:energya}, one will perform \emph{smoothing} by using sudden bursts of energy; thereby approximating the constant velocity Spacetime path \ref{def:path1}.  In this case, there are two legs in the journey (taking time $T_1=T_2 = 1/2N^{1/n}$ sec in $O$'s frame) and the energy required on each leg is half the total energy.  Consider again the event \textbf{0} from Section~\ref{sec:constantV} where we start the classical and quantum computers at $t=\tau=0$ in the same inertial frame and fix $n>2$ from \eqref{eqn:Econstant}.  One will return to find that the classical computer outperformed the quantum one, on the same black-box search problem.

\section{Conclusion}
This study is part of the research effort aimed at understanding what class of computations are made possible or ruled out by the laws of physics~\cite{De85,lloyd-2000,Yao03,aar:np}.  We have shown that finite $n$th root polynomial reductions in algorithmic run-time are made possible by relativistic effects.  The runtime improvement is predicted by Einstein's theory of relativity and the connection to computation was explored by considering polynomial reductions inside the black-box model.  In the present study, the observer changes her own life history by taking an accelerated spacetime path to obtain a computational time efficiency improvement.  The method would be more practical if instead an inertial observer sent the computer on an alternative spacetime path\footnote{This is possible when considering gravitational time-dilation.  For instance, if $gh = \Phi_A - \Phi_B$ is the gravitational potential between two observers at rest ($O$ at sea level, and $I$ at some higher elevation $h>0$), than the clocks relate as $t = \tau(1-gh/c^2)$, where terms of order $(gh/c^2)^2$ are assumed to be negligible.}.  

\begin{acknowledgements}
I thank Ettore Minguzzi, Raymond Lal, Peter Love, James Whitfield, Piotr Chrusciel, Stephen Jordan, Mark Williamson and acknowledge financial support from EPSRC grant EP/G003017/1.
\end{acknowledgements}

\end{document}